\documentclass[11pt]{article}
\usepackage[dvips]{graphicx}  
\usepackage{amsmath} 
\usepackage{tabularx}
\usepackage{setspace}   
   
\newcommand{\ba}{\begin{align}}  
\newcommand{\ea}{\end{align}}

\usepackage{amsfonts}   
\usepackage{amssymb} 
\usepackage{amsthm}  

\newtheorem{thm}{Theorem}[section]

\newtheorem{lemma}[thm]{Lemma}

\newtheorem{rmk}[thm]{Remark}
\numberwithin{equation}{section}

\def\be{\begin{eqnarray}} 
\def\ee{\end{eqnarray}}
\def\bee{\begin{eqnarray*}}
\def\eee{\end{eqnarray*}}
\def\bmx{\begin{pmatrix}}
\def\emx{\end{pmatrix}}

\def\ts{\textstyle}

\def\rt2{\ts \frac{1}{\sqrt{2}} }

\newcommand{\N}{\mathbb{N}}
\newcommand{\Z}{\mathbb{Z}}

\newcommand{\R}{\mathbb{R}}
\newcommand{\C}{\mathbb{C}}

\DeclareMathOperator{\trace}{Tr}

\DeclareMathOperator*{\Ex}{\mathbb{E}}

\title{Revisiting additivity violation of quantum channels}     

\author{Motohisa Fukuda
\footnote{e-mail: m.fukuda@tum.de,  address: Zentrum Mathematik, M5,
Technische Universit\"at M\"unchen, 
Boltzmannstrasse 3,
85748 Garching, Germany}  
}

\begin{document}  

\maketitle

\begin{abstract}
We prove additivity violation  of minimum output entropy of quantum channels  by straightforward application of $\epsilon$-net argument and L\'evy's lemma. 
The additivity conjecture was disproved initially by Hastings. Later, a proof via asymptotic geometric analysis was presented by Aubrun, Szarek and Werner,
which uses Dudley's bound on Gaussian process (or Dvoretzky's theorem with Schechtman's improvement). 
In this paper, we develop another proof along Dvoretzky's theorem in Milman's view  
showing additivity violation in broader regimes than the existing proofs. 
Importantly, 
Dvoretzky's theorem  works well with norms to give strong statements but 
these techniques can be extended to functions which have norm-like structures
- positive homogeneity  and triangle inequality.  
Then,  a connection between Hastings' method and ours is also discussed. 
Besides, we  make some comments on relations between regularized minimum output entropy and classical capacity of quantum channels. 
\end{abstract}

\section{Preliminary}    
\subsection{Introduction} 
Existence of quantum channels which show
additivity violation of minimum output entropy was proven by Hastings \cite{Hastings2009}, which is stated as follows. 
There exists a quantum channel $\Phi$ such that,  denoting  the complex conjugate of $\Phi$ by $\bar \Phi$,
\be\label{eq:violation}
S_{\min} (\Phi \otimes \bar \Phi) < S_{\min} (\Phi) + S_{\min} (\bar \Phi)
\ee
Here, $S_{\min}(\cdot)$ is the \emph{minimum output entropy} of quantum channel, which  is defined as 
\be
S_{\min} (\Phi) = \min_{\rho} S(\Phi(\rho))
\ee
where $S(\cdot)$ is the \emph{von Neumann entropy} and $\rho$ runs over all the \emph{pure states} which are rank-one projections;
general quantum states (\emph{mixed states}) are written by positive Hermitian matrices of trace one but
we can assume that input states are pure states because the function $S(\cdot)$ is concave.  
Historically, the additivity question was made in \cite{KingRuskai2001}. 
Note that $\leq$ is obvious  in \eqref{eq:violation}.
 
One of  consequences of the additivity violation of minimum output entropy is that  \emph{Holevo capacity} is not in general additive either.  
Holevo capacity $\chi(\cdot)$ is defined as
\be\label{Holevo-cap} 
\chi(\Phi) = \max_{\{p_i,\rho_i\}} \left [S\left( \Phi\left(\sum_i p_i\rho_i \right)\right) - \sum_i p_i S\left(\Phi\left( \rho_i \right)\right)  \right]
\ee
Here, $\{p_i,\rho_i\}$ runs over all possible ensembles, where probabilities $\{p_i\}$ are assigned to quantum states $\{\rho_i\}$
\cite{SchumacherWestmoreland1997,Holevo1998}. 
One form of additivity violation of Holevo capacity can be stated as 
\be
\chi\left(\Phi^{\otimes 2}\right)   > 2 \chi(\Phi)
\ee
This is deduced  by getting the additivity violation of minimum output entropy for two identical channels from \eqref{eq:violation} 
via the the result in \cite{FW2007}, and 
using the equivalence relation between the two additivity or non-additivity properties \cite{Shor2004}. 
(The latter technique is extended in Section \ref{sec:reg} to show a similar statement for regularized quantities.) 
This, in turn, implies that
\be\label{cc}
C(\Phi) = \lim_{n \to \infty} \frac 1n \chi\left( \Phi^{\otimes n}\right)  =  \lim_{m \to \infty} \frac 1{2m} \chi\left( \Phi^{\otimes 2m}\right) 
\geq \frac 12 \chi\left( \Phi^{\otimes 2}\right)  > \chi (\Phi) 
\ee
Here, $C(\cdot)$ is the classical capacity and this operational quantity is defined in an asymptotic form as in the first equality.   
On the other hand, Holevo capacity is written in a one-letter formula as in \eqref{Holevo-cap} and it gives the classical capacity when we do not use entangled inputs
\cite{SchumacherWestmoreland1997,Holevo1998}. 
These two quantities had been conjectured to be identical but now we know that they are different in general as one can see in \eqref{cc}.
Therefore, this in particular implies that entanglement inputs can increase the classical capacity of some channels. 
We refer interested readers to \cite{Holevo2006}.

Soon after the Hasting's paper \cite{Hastings2009} was publicized in 2008, several papers followed to give rigorous proofs and generalize the result
\cite{FKM2010, BrandaoHorodecki2010, FK2010}. 
Moreover, in 2010, Aubrun, Szarek and Werner found another proof in \cite{AubrunSzarekWerner2011} 
via the Dudley's bound on Gaussian process \cite{Dudley1967,JainMarcus1978}
(or Dvoretzky's theorem with Schechtman's improvement \cite{Schechtman1989}) .
The original Dvoretzky's theorem can be found in \cite{Dvoretzky1961}.
In fact, a  year before, they proved in \cite{AubrunSzarekWerner2010} additivity violation of $p$-Renyi entropy for $p>1$ 
via Dvoretzky's theorem in Milman's version \cite{Milman1971,FigielMilman1977}, but
it was not strong enough to prove additivity violation of minimum output entropy as was written in \cite{AubrunSzarekWerner2011}.
(The additivity violation of $p$-Renyi entropy for $p>1$ itself was first proven by Hayden and Winter in 2007 \cite{HaydenWinter2008},
and later by Collins and Nechita \cite{CollinsNechita2011} via  \emph{free probability}.)
Also, Additivity violation for $p$ close to $0$ was proven in \cite{CHLMW2008}. 
Note that our problem corresponds to the case $p=1$. 
Interestingly, no concrete counterexample has  been found yet for $p=1$ whereas
a counterexample for $p>2$ was explicitly constructed in \cite{GHP2010}, many years after
the counterexample for $p>4.79$ was found in \cite{WernerHolevo2002}.
Also, we must mention a recent paper \cite{BCN2013} where they proved a rather large additivity violation and the smallest output dimension could be as small as 183 while the dimensions of input and environment are infinite. Their method is based on free probability.

In this paper, we show that additivity violation of minimum output entropy can be proven 
by the standard method via \emph{$\epsilon$-net argument} and L\'evy's Lemma, which in fact is very similar to the Milman's view on Dvoretzky's theorem.
Interestingly, this pair of techniques was used in \cite{HLW2006} to show existence of strongly entangled subspaces, which finally
lead to the additivity violation of Ranyi entropy for $p>1$ \cite{HaydenWinter2008}.
However, their estimate was not strong enough to prove the additivity violation of minimum output entropy. 
On the other hand, our new approach gives an improved estimate which makes it possible. 
Historically, 
approximating the von Neumann entropy by using the Hilbert-Schmidt distance from the maximally mixed state, which was introduced in \cite{BrandaoHorodecki2010}
(perhaps originally from \cite{Hastings2009} via Taylor expansion),
 fitted into asymptotic geometric analysis argument  in \cite{AubrunSzarekWerner2011}. 
Moreover we suggest that its norm-like structures - almost positive homogeneity and triangle inequality - actually put our problem into 
the framework of the Milman's view. 
The technical discussion on this issue is written in Section \ref{sec:tech} after stating additivity violation in Section \ref{sec:main}.
One can see that our result is stronger than all the existing proofs \cite{Hastings2009,FKM2010, BrandaoHorodecki2010, FK2010, AubrunSzarekWerner2011}
 in the sense that we can prove the additivity violation asymptotically as long as the dimensions of input and output are proportional to each other
and proportionally larger than or equal to square of the dimension of environment; there is no restriction on the ratios. 
We make some analysis on our method and compare  it to Hastings'  in Section \ref{sec:Hastings}.
Our proof method can be applied to random unitary channels, which is briefly studied in Section  \ref{sec:RU}.

Besides, there is an open problem:   
\be\label{open} 
C(\Phi \otimes \Omega) \stackrel{?}{=} C(\Phi) + C(\Omega)  
\ee
for different channels $\Phi$ and $\Omega$. 
In Section \ref{sec:reg}, we provide a proof with the widely-known fact that 
additivity violation of regularized minimum output entropy implies that of classical capacity.  

\subsection{Channel}
A (quantum) state is a positive Hermitian operator of trace one, and a (quantum) channel is a completely positive and trace-preserving map on the states. 
We denote the set of unit vectors in  $\C^d$ by  $S_{\C^d}$ and
the linear maps on $\C^d$ by  $ L(\C^d)$.
Also, we denote the dual of vector $x \in \C^d$ by $x^*$, where,
in the \emph{bra-ket} notation, $x = |x\rangle$ and $x^* = \langle x|$,
which we don't use in this paper. 

In Stinespring's picture \cite{Stinespring1955}, channels are identified as isometries:
\be
V: \C^l \to \C^k \otimes \C^n 
\ee
and channels are written for $x \in S_{\C^l}$ as
\be
\Phi: L\left(\C^l\right)&\to& L \left(\C^k\right) \\ 
xx* &\mapsto& \trace_{\C^n}  [Vxx^*V^*]
\label{channel}
\ee
Moreover, through this embedding picture,  
we can identify quantum channels as $l$-dimensional subspaces $E \subseteq \C^k \otimes \C^n$ such that 
\be
\Phi_E : L(E) &\to& L\left(\C^k\right)  \\
xx^* &\mapsto&  \trace_{\C^n}  [xx^*]=XX^*
\label{channel2}
\ee
for  $x \in \tilde E = E \cap S_{\C^k \otimes \C^n}$. 
Here,  partial trace is understood by the following identification between linear spaces: 
\be
\C^k \otimes \C^n &=& \mathcal M_{k,n}(\C) \\
x &=& X \label{vec-mat}
\ee
In what follows, we use the lower and upper cases of the same letter to represent this identification in \eqref{vec-mat}. 
Importantly, by applying Schmidt decomposition to $x \in \C^k \otimes \C^n$ we know that
 $\trace_{\C^n}  [xx^*]$ and $\trace_{\C^k}  [xx^*]$ share the same non-zero eigenvalues.
So, we always assume safely that $k \leq n$. 
In this case, we say that the dimensions of input, output and environment are $l$, $k$ and $n$, respectively,
although the spaces of environment and output are interchangeable for the additivity problem of minimum output entropy \cite{Holevo2005,KMNR2007}.   

We give some definitions here. 
To define the complex conjugate of channel $\Phi_E$, 
 we fix some isometry $V$  in \eqref{channel} such that its image is $E$ and then
define the channel $\bar \Phi_E$ by $\bar V$. 
This definition is unique only up to rotations, but this does not cause a problem in our paper. 
Since we identify channels as subspaces we define random quantum channels as follows. 
Fix an $l$-dimensional subspace $E_0$ in $\C^k \otimes \C^n$ and generate random subspaces 
$UE_0$ with $U \in \mathcal U(kn)$ where $U$ is chosen randomly according to the Haar measure on the unitary group.

\section{Additivity violation} \label{sec:main}

First, the canonical Bell state on $\C^d \otimes \C^d$ is defined as  
\be
b_d = \frac 1 {\sqrt d} \sum_{i=1}^d e_i \otimes e_i 
\ee
where $\{e_i\}$ is the canonical basis in $\C^d$.
The Bell state flips matrices with transpose and in particular: 
\be\label{Bell-p}
\left(U \otimes \bar U\right) b_d = \left(U \left( \bar U \right)^T \otimes I \right) b_d   = b_d 
\ee
for $U \in \mathcal U(d)$.  
This property ensures a rather large eigenvalue of $(\Phi_E \otimes \bar \Phi_E) (b_lb_l^*)$.
This idea is originated from \cite{HaydenWinter2008} where $l$ divides $kn$ but 
their proof is easily adapted to any $l \leq kn$, which is written below.
(A proof on this property through \emph{graphical calculus} was given in \cite{CollinsNechita2010}, in which
 the exact limit eigenvalue distribution of $(\Phi \otimes \bar \Phi) (b_lb_l^*) $ was also calculated with random isometry in the picture of \eqref{channel}.)
This, in turn, implies the following lemma.
\begin{lemma}\label{Bell}
For any channel $\Phi_E$, let $l = an$ with $a >0$ and we have    
\be
S_{\min} \left(\Phi_E \otimes \bar \Phi_E \right) \leq  2 \log k - \frac {a \log k} k + \frac {2a} k   
\ee  
for large enough $k$. 
\end{lemma}  
\begin{proof} 
Firstly, we get a lower bound for the largest eigenvalue of 
\be
(\Phi_E \otimes \bar \Phi_E) (b_lb_l^*) = \trace_{\C^n \otimes \C^n} \left[ \left(V \otimes \bar V\right)b_l b_l^* \left(V^* \otimes V^T \right) \right] 
\ee
by projecting them to the one-dimensional subspace of the Bell state as in \cite{HaydenWinter2008}:
\be
b_k^*(\Phi_E \otimes \bar \Phi_E) (b_lb_l^*)  b_k  
\geq \left| b_l^* \left(V^* \otimes  V^T \right) \left( b_k\otimes  b_n\right)  \right|^2
=  \frac{l}{kn}  \left| b_l^* b_l \right|^2 = \frac{l}{kn}
\ee
Indeed, for $l \times d$ matrix $A$,
\be
(I_d \otimes A) \sum_{i=1}^d e_i \otimes e_i  = \sum_{i=1}^d e_i \otimes \left(\sum_{j=1}^l A_{j,i} f_j \right) 
=\sum_{j=1}^l  \left(\sum_{i=1}^dA_{j,i} e_i \right) \otimes  f_j = \left(A^T \otimes I_l \right) \sum_{j=1}^l  f_j \otimes  f_j
\ee
where $\{f_j\}_{j=1}^l$ is the canonical basis in $\C^l$.

Secondly, the bound in the statement of theorem is derived from the largest possible entropy under this constraint.   
\be
-\frac{a}{k} \log \left( \frac a k \right) - \left( 1-\frac a k \right) \log \left( \left( 1 - \frac a k \right)  \frac{1}{k^2-1} \right)   \notag
&\leq & \frac a k \log k -\frac a k \log a + \left(  1 - \frac a k  \right) \left [ 2 \log k - \log \left( 1- \frac a k  \right)    \right]  \\
&\leq&2 \log k - \frac {a \log k} k + \frac a k  \left[2 - \log a- \frac {2a} k \right]
\ee
for large enough $k$ so that the bound $\log(1-\frac a k) \geq - \frac{2a} k$ holds. 
\end{proof}

Next, 
the following approximating bound of the von Neumann entropy around the maximally mixed state was introduced in \cite{BrandaoHorodecki2010}. 
\begin{lemma} \label{lemma:approx}
For any state $\rho$ on $\C^k$, 
\be 
\log k - S(\rho)  \leq  k \cdot \left\| \rho - \tilde I_k \right\|_2^2
\ee
where $\tilde I_k = I_k/k$, the identity on $\C^k$ normalized to be trace-one.
\end{lemma} 
The bound can be seen easily from the concavity of $S(\cdot)$ in particular around the maximally mixed state. 
This idea extremely fits into asymptotic geometric analysis as was pointed out in \cite{AubrunSzarekWerner2011}. 
Also, it fits even better to our method because the function in \eqref{function:f}, which is made out of Lemma \ref{lemma:approx},  almost shows
positive homogeneity and triangle inequality. 

We are now ready to state the main theorem:
\begin{thm}[Main theorem]\label{thm:main}
Suppose  $l =an$ and $k^2 = \beta n$ for $a,\beta>0$, i.e., $l ,n\sim k^2$. Then,
 we observe additivity violation of minimum output entropy when $k$ is large enough.  
Moreover, the statement holds even if $\beta \to 0$, i.e. in the regime where $l \sim n$ and $k^2 = O(n)$. 
\end{thm} 
\begin{proof}
Take $\theta$ and $\epsilon$ as in Theorem \ref{exist-subspace}.
For example, 
set $\theta = 1/4$ and 
\be
\epsilon =  2 \sqrt{ a \log \left( 1 + \frac 2 \theta \right)}
\ee
Then,  Theorem \ref{exist-subspace} implies that  
there exists some constant $C>0$ and subspace $E$ such that   
\be
\max_{x \in \tilde E} \left\|\Phi_E (xx^*) - \tilde I_k \right\|_2 \leq \frac C k
\ee
for sufficiently large $k$. Here, $\tilde E = E \cap S_{\C^k \otimes \C^n}$.
Therefore, by using Lemma \ref{Bell} and Lemma \ref{lemma:approx} we have 
\be\label{eq:last}
S_{\min}\left (\Phi_E \otimes \bar \Phi_E \right) \leq 2 \log k - \frac {a \log k} k + \frac {2a} k 
< 2 \left [  \log k - \frac {C^2} k \right] \leq S_{\min}\left (\Phi_E\right) + S_{\min} \left(\bar \Phi_E\right)  
\ee
for large enough $k$. 
Note that $S_{\min} \left(\bar \Phi_E\right)  =S_{\min} \left( \Phi_E\right)  $. 
\end{proof} 

\begin{rmk}
By using complementarity, we can think of $\mathbb C^n$ as the output space. If we in addition suppose that $l=n$, 
the input and output dimension is $n$ and the environment dimension $k$. 
Then, Theorem \ref{thm:main} implies that the additivity violation happens in the regime where
$n \gtrsim k^2$. 
In particular, if $k$ is fixed to satisfy \eqref{eq:last}, one can take  $n$ large enough to get the additivity violation.  
\end{rmk}

\section{Technical part}\label{sec:tech}

Define a function $f: \C^k \otimes \C^n \to \R$ as
\be\label{function:f}
f(x) = \left\|XX^* - \trace [XX^* ]\tilde I_k \right\|_2  
\ee 
Here, again we use the upper and lower cases to show the identification in \eqref{vec-mat}.
The following lemma describes the typical behavior of $f(x)$ for $x \in S_{\C^k \otimes \C^n}$
when it is chosen uniformly random. 
The result shows  why we need $n \gtrsim k^2$ for additivity violation in our framework for general quantum channels, and
it seems impossible to improve the orders.
\begin{lemma}  \label{med-bound}
For  $x \in S_{\C^k \otimes \C^n}$  uniformly distributed, 
\be
\Ex_x [f]\leq \frac1  {\sqrt{n}}  , \qquad
\mathrm{med}(f) \leq \frac 1 {\sqrt{n}}  \left( 1+ \frac 3 {\sqrt{k}}\right) 
\ee
Here, $\Ex(\cdot)$ and $\mathrm{med}(\cdot)$ are mean and median respectively. 
\\
Comment: for our main theorem, one only needs  $\mathrm{med}(f) \lesssim k^{-1}$ assuming  $n \gtrsim k^2$. 
\end{lemma} 
\begin{proof}
By the Jensen's inequality
\be 
\left( \Ex \left[ f(x) \right] \right)^2  \leq \Ex \left[ (f(x))^2\right] =  \Ex \left[\trace\left[\left(XX^*\right)^2\right]\right] - \frac 1 k
\ee 
Further, we see that there exist $\alpha, \beta >0$ and
\be
\Ex \left[\trace\left[\left(XX^*\right)^2 \right] \right] 
=  \trace \left [\Ex [xx^* \otimes xx^*]( P_k \otimes I_{n^2})  \right]   = \trace [(\alpha I_{k^2n^2} + \beta P_{kn}) (P_k \otimes I_{n^2})]
\ee
where $P_k$ and $P_{kn}$ are swapping matrices on $\C^k \otimes \C^k$ and $(\C^k \otimes \C^n)\otimes(\C^k \otimes \C^n)$, respectively. 
Here, for the last equality we used the Schur's lemma as in \cite{BrandaoHorodecki2010} because the expectation is invariant for $U \otimes U$ with $U \in \mathcal U (kn)$. 
Since $\alpha =\beta = \frac 1 {kn(kn+1)}$ we get a bound:
\be
 \Ex \left[ f(x) \right] \leq \sqrt{\frac{k+n}{kn+1}- \frac 1 k} \leq \frac1  {\sqrt{n}}  
\ee

For the second statement, the standard argument proceeds as follows. 
\be
\left| \Ex (f) - \mathrm{med} (f) \right|
&\leq& \Ex \left| f  - \mathrm{med} (f) \right|  \\
&\leq& 2 \sqrt{\frac{\pi}{8}} \int_0^{+\infty} 
\exp \left\{ - \frac{\left(kn- 1\right)\varepsilon^2} 4  \right\} \,d \varepsilon 
= \sqrt{\frac{\pi}{8}}  \cdot \sqrt{ \frac {4 \pi}{kn - 1}}
 \leq \frac 3  {\sqrt{kn}}
\ee
Here, we applied Jensen's inequality and Lemma \ref{Levy} for the first two inequalities. 
Note that in this calculation we set $L=2$ in \eqref{eq:Levy}, 
by using the bound \eqref{L-bound} with $\|X\|_\infty, \|Y\|_\infty \leq 1$,
which works for all $\varepsilon>0$.
The second last equality comes from the identity for the Gaussian distribution. 
\end{proof}  
 
All the existing papers on additivity violation of minimum output entropy via measure concentration argument
use large deviation bounds similar to the one in Theorem \ref{deviation}.
Especially the order $\exp\{-n\}$, instead of $\exp\{-n/k\}$, is important for their proofs. 
To this end, 
they essentially rectify the concerned functions on the unit spheres and apply the L\'evy's lemma.
However, in this rectifying process, one needs another large deviation bound from random matrix theory
or some extra efforts.  
In Theorem \ref{deviation}, we avoid this complication and prove the desired bound 
directly via the L\'evy's Lemma (see Lemma \ref{Levy}).
For this purpose, we need a small lemma before going on to the theorem:
\begin{lemma}\label{lemma:trick8}
Let $x \in S_{\C^k \otimes \C^n}$ be such that $f(x) \leq \mathrm{med}(f)$. Then,
\be
\|X\|_\infty  \leq   \frac 1 {\sqrt{k}} +2 \sqrt{\frac k n}
\ee
\end{lemma} 
\begin{proof}
The condition $f(x) \leq \mathrm{med}(f)$ implies via Lemma \ref{med-bound} that
\be
\|XX^*\|_\infty \leq \left\|XX^*- \tilde I_k \right\|_2 + \left\| \tilde I_k \right\|_\infty  
\leq \frac 1 {\sqrt{n}} \left( 1 + \frac 3 {\sqrt k} \right)  + \frac 1 k 
\leq \left( \frac 1 {\sqrt{k}} +2 \sqrt{\frac k n} \right)^2 
\ee 
\end{proof} 

This lemma gives the desired large deviation bound in a straightforward way:
\begin{thm} \label{deviation}
Let $k^2 = \alpha^2 n$ with $\alpha >0$. Then, 
for  $x \in S_{\C^k \otimes \C^n}$  uniformly distributed, 
\be
\Pr \left\{  f(x)  > h(k,\alpha,\epsilon) + \mathrm{med}(f) \right\}   
< \sqrt{\frac \pi 8}\exp \left\{ - \epsilon^2 \left( n- \frac 1k \right) \right\} 
\ee
for all $\epsilon>0$.  
Here, 
\be
 h(k,\alpha,\epsilon)  = \frac  {2 \epsilon (1  + 2\alpha + \epsilon) } k 
\label{function:h}
\ee
\end{thm} 
\begin{proof}
We follow the notations in Theorem \ref{Levy}.  
Let $A= \{x \in S_{\C^k \otimes \C^n}: f(x) \leq \mathrm{med}(f) \}$ and then Lemma \ref{lemma:trick8}  shows that 
 for $x \in A^\varepsilon$ with $\varepsilon = \frac{\epsilon}{\sqrt{k}}$, we have
\be \label{infinity-bound}
\|X\|_\infty \leq 
\frac  {1  + 2\alpha + \epsilon}{\sqrt{k}} 
\ee
Here, we used the fact that $\|\cdot\|_\infty \leq \|\cdot\|_2$.
Hence, we can set an upper bound of the Lipschitz constant on $A^{\epsilon /\sqrt k}$ to be  twice as large as \eqref{infinity-bound}. 
Indeed, for $x,y \in \C^k \otimes \C^n$, 
\be \label{L-bound}
| f(x) - f(y)| \leq \left\| XX^* - YY^* \right\|_2 \leq \left( \|X\|_\infty + \|Y\|_\infty \right) \left\| X-Y \right\|_2  
\ee
This trick on the Lipschitz constant with $\|\cdot\|_\infty$ was used in \cite{AubrunSzarekWerner2011}
and originally from \cite{BrandaoHorodecki2010}.
Therefore, applying Lemma \ref{Levy} completes the proof;
$\varepsilon L$ in \eqref{eq:Levy} is replaced by
\be
\frac \epsilon {\sqrt k} \cdot  \frac  {2( 1  + 2\alpha + \epsilon)}{\sqrt{k}} 
\ee
which is what we want as $h(k,\alpha,\epsilon)$. 
\end{proof} 
 
The following lemma brings our problem back to Milman's view of Dvoretzky's theorem.
We define a $\theta$-net to be a subset of the unit sphere such that 
any point on the sphere finds a point in the subset within distance $\theta$. 
This approximation technique works well not only with norms, as in Milman's view, but also
with functions having more or less positive homogeneity and triangle inequality:
\begin{lemma}\label{lemma:net}
Let $E$ be an $l$-dimensional subspace in $\C^k \otimes \C^n$    and   $\tilde E$  the unit sphere in it. 
Then, we can construct a $\theta$-net on $\tilde E$, denoted by $N_\theta$, so that the following statements hold. 
\be
|N_\theta| &\leq& \left( 1 + \frac 2 \theta \right) ^{2l} \\
\max_{x \in \tilde E} f(x) &\leq& \frac 1 {1-\theta^2- 2\theta } \cdot \max_{x \in N_\theta} f(x)
\label{net-bound}
\ee 
for $\theta >0$ such that RHS of \eqref{net-bound} is positive, in particular $0<\theta\leq \frac 14$. 
\end{lemma} 
\begin{proof}
Since the first bound is well-known, for example see \cite{PisierBook}, we only prove the second statement. 
For any $v \in \tilde E$ there exists $x\in N_\theta$, $y\in \tilde E$ and $0\leq \delta \leq\theta$ such that $v=x+\delta y$. 
Then, 
\be
f(x+\delta y) &\leq& \left\|XX^* - \tilde I_k \right\|_2 + \delta^2  \left\| YY^* - \tilde I_k \right\|_2 + 
\delta \underbrace{ \left\| XY^* + YX^* - \trace \left[ XY^* + YX^*\right] \tilde I_k \right\|_2 }_{(\star)}\\
&\leq&\max_{x \in N_\theta} f(x) + \left(\delta^2 +2\delta \right) \cdot  \max_{x \in \tilde E} f(x)     
\ee
Indeed, since $\trace_{\C^n} [xy^*+yx^*] = XY^* + YX^*$, we firstly write 
\be 
xy^* + yx^* = \alpha\, zz^* + \beta\, ww^* 
\ee
for some orthonormal $z,w \in \tilde E$ and $\alpha,\beta \in \R$, and secondly, we get
\be
(\star) &\leq& |\alpha| \left\| ZZ^* - \trace \left[ZZ^*\right] \tilde I_k \right\|_2 + |\beta| \left\| WW^* - \trace \left[WW^*\right] \tilde I_k \right\|_2 \\
&\leq& \left( |\alpha| + |\beta|  \right) \max_{x \in \tilde E} f(x)  \leq  2 \max_{x \in \tilde E} f(x)  
\ee
Here, we used the following bound:
\be
 |\alpha| + |\beta| = \left\| xy^* + yx^* \right\|_1 \leq \left\| xy^* \right\|_1 + \left\|  yx^* \right\|_1=2
\ee
This completes the proof. 
\end{proof} 

\begin{thm}\label{exist-subspace}
Suppose we have random $l$-dimensional subspaces $E \subset \C^k \otimes \C^n$ where $k^2 = \alpha^2 n$. 
For any $\epsilon>0$ and $0<\theta\leq \frac 14$, if we choose $l$ such that
\be \label{condition-l}
\frac {l} n \leq \frac{\epsilon^2 }{4 \log \left(1+\frac 2 \theta\right)}  
\ee
then there exists a subspace $E$ such that  
\be
\max_{x \in \tilde E} f(x) \leq  \frac 1 {1-\theta^2-  2\theta  } \cdot \left[h(k,\alpha,\epsilon) +  \frac {4\alpha} k \right]
\label{thebound}
\ee
where the function $h(\cdot,\cdot,\cdot)$ is defined in \eqref{function:h}.
\\
Comment: 
An important message of this theorem is that the RHS of \eqref{thebound} is bounded by $\frac{C(\alpha,\theta,\epsilon)} k$  where
$C(\alpha,\theta,\epsilon)$ is some constant depending on $\alpha$, $\theta$ and $\epsilon$.
\end{thm} 
\begin{proof}
Fix a subspace $E_0$ of dimension $l$ and construct a $\theta$-net on $\tilde E_0 = E_0 \cap S_{\C^k \otimes \C^n}$,
which we denote by $N_\theta$.
We calculate, 
\be
&&\Pr_{U \in \mathcal U (kn)} \left\{ f(x) > h(k,\alpha,\epsilon) +\frac {4\alpha} k   , \quad \exists x \in U N_\theta   \right\}  \\
&\leq&  |N_\theta| \cdot \Pr_{U \in \mathcal U (kn)}  \left\{ f(Ux_0) > h(k,\alpha,\epsilon) + \mathrm{med}(f) , \quad \text{ for fixed $x_0 \in  N_\theta  $} \right\}  \\
&\leq& \exp \left\{ 2l \log \left( 1+ \frac 2 \theta \right)  \right\} 
\times  \sqrt{\frac \pi 8}\exp \left\{ - \epsilon^2 n\left(1- \frac 1{kn}\right)  \right\} 
\label{prob-bound} 
\ee
Here, we used the first statement of Lemma \ref{lemma:net} and Theorem \ref{deviation}.
Since $1-\frac 1{kn}>\frac 1 2$, the condition \eqref{condition-l} implies that \eqref{prob-bound} is smaller than one.
Hence there exists $U \in \mathcal U(kn)$ such that 
\be
\max_{x \in U N_\theta} f(x) \leq   h(k,\alpha,\epsilon) + \frac {4\alpha} k 
\ee
Therefore, for this $U$, set $E=UE_0$ where $UN_\theta$ constitutes a $\theta$-net for $\tilde E = U \tilde E_0$ so that
the second statement of  Lemma \ref{lemma:net} completes the proof. 
\end{proof} 

The above Theorem \ref{exist-subspace} is ``tailored'' to prove the additivity violation.
Below, we give a similar statement in a different view point, 
which gives an upper bound for the function $f(\cdot)$ on a typical $l$-dimensional subspace of $\mathbb C^k \otimes \mathbb C^n$, and
which we believe brings better understanding of Theorem \ref{exist-subspace}. 
\begin{thm}\label{thm:typical-l}
For $l,k,n \geq 2$, there exists a subspace $E \subset \mathbb C^k \otimes \mathbb C^n$ of dimension $l$ such that 
\be 
\max_{x \in \tilde E} f(x) 
\leq 15 \frac 1 k \sqrt {\frac l n} + 30 \frac {\sqrt l} n + 36 \frac l {kn} + 10 \frac 1 {\sqrt n}
\ee
Note that the above integer constants are not chosen to be tight. 
\end{thm}
\begin{proof}
As in the proof of Theorem \ref{exist-subspace}, we calculate the following probability:
\be
&&\Pr_{U \in \mathcal U (kn)} \left\{ f(x) > \Delta \cdot 2 \left( \frac 1 {\sqrt k} + 2 \sqrt{\frac k n} + \Delta \right) + \frac 1 {\sqrt n} \left( 1 + \frac 3 {\sqrt k} \right)  , \quad \exists x \in U N_\theta   \right\}  \\
&\leq&  |N_\theta| \cdot \Pr_{U \in \mathcal U (kn)}  \left\{ f(Ux_0) 
> \Delta \cdot \underbrace{2 \left( \frac 1 {\sqrt k} + 2 \sqrt{\frac k n} + \Delta \right) }_{(*)}+ \mathrm{med}(f) , 
\quad \text{ for fixed $x_0 \in  N_\theta  $} \right\}  \\
&<& \exp \left\{ 2l \log \left( 1+ \frac 2 \theta \right)  \right\} 
\times \exp \left\{ - \Delta^2 (kn-1) \right\} 
\ee
Here, $\Delta>0$ and $(*)$ correspond to  $\varepsilon$ and $L$ in \eqref{eq:Levy}. 
To make the probability smaller than one, we need 
\be
\Delta^2 \geq \frac {2l \log (1 + \frac 2 \theta)} {kn-1}
\ee
Hence, set $\theta = \frac 1 4$, for example, and
\be
\Delta = \sqrt{\frac{6l}{kn}} > \sqrt{ \frac {2l \log (1 + \frac 2 \theta)} {kn-1} }
\ee
Then, there exists a subspace $E$ such that
\be
\max_{x \in \tilde E} f(x) 
&\leq& 3\cdot \left[\sqrt{\frac{6l}{kn}} \cdot 2 \cdot 
\left( \frac 1 {\sqrt{k}} + 2 \sqrt{\frac k n} + \sqrt{\frac{6l}{kn}} \right) + 
\frac1 {\sqrt n} \left( 1+  \frac 3 {\sqrt k}\right) \right] \\
&\leq& 6\sqrt 6 \frac 1 k \sqrt {\frac l n} + 12 \sqrt 6 \frac {\sqrt l} n + 36 \frac l {kn} + \left(3+ \frac 9 {\sqrt 2}\right) \frac 1 {\sqrt n}
\ee
This completes the proof. 
\end{proof}

\section{Hastings' proof and ours}  \label{sec:Hastings}
An important step in our proof can be seen in \eqref{net-bound} where  
the bound over the whole domain (subspace)  can be set to be, for example, twice as large as  the bound only over the net
if one properly chooses $\theta>0$. 
We emphasize here that choice of $\theta$ is independent of $k$.
If we had thought of this problem by using the Lipschitz constant, the correction would be an additive term instead of a multiplicative constant. 
Since the Lipschitz constant is at best proportional to $\frac 1 {\sqrt{k}}$, the additive correction  would be proportional to $\frac \theta {\sqrt k}$. 
However, we need a bound proportional to $\frac 1 k$. Hence $\theta$ must be proportional to $\frac1 {\sqrt k}$,
which would give an unwanted $k$-dependent factor in  \eqref{condition-l}.  
Therefore it is crucial in our method to use ``positive homogeneity and triangle inequality'' of function $f$ in order to get the bound    \eqref{net-bound}.  

In this kind of problems, one of useful approaches is to ask  ``how much of the domain can be approximated by one point''. 
In our proof, it is 
$\exp \left\{ -2 n \log \left( 1+ \frac 2 \theta \right)  \right\}$ when $l=n$. 
We dare to say that this corresponds to (37), derived from (34), in the supplementary information of \cite{Hastings2009}.
This idea of him is roughly stated as follows, hoping that there is not misunderstanding.   

One can decompose uniformly distributed  $z \in S_{\C^n}$ as
\be
z = \omega x + \sqrt{1-|\omega|^2} y
\ee
where $x$ is fixed  and $y$ is uniformly distributed on $S_{x^\perp} \simeq S_{\C^{n-1}}$. 
Also, note that $|\omega|^2$ has the law of Beta distribution. 
Via this decomposition, we have
\be
\Phi(zz^*) \approx |\omega|^2 \Phi(xx^*)  + (1-|\omega|^2) \Phi(yy^*) \approx |\omega|^2 \Phi(xx^*)  + (1-|\omega|^2) \tilde I_k
\label{Hastings-deco}
\ee
The second approximation is assumed because generically channels send random inputs to a neighborhood of $\tilde I_k$ 
although we need a careful analysis for this statement. For example, see \cite{FKM2010},
where the important idea \emph{tubal neighborhood} was reformulated as \emph{TUBE}. 
However we believe that we arrive at the same goal, or at least get convinced, 
if we look at \eqref{Hastings-deco} in the Hilbert-Schmidt norm. 
First, we have 
\be
\underbrace{ \left\| \Phi(zz^*) - \tilde I_k \right\|_2}_{(\star)} \gtrapprox|\omega|^2  \underbrace{ \left\| \Phi(xx^*) - \tilde I_k \right\|_2 }_{(*)}
\ee
but  $|\omega|^2>\frac 12$ occurs with probability $\exp \{-n \log 2\}$ because 
$|\omega|^2$ has the law of the beta distribution $B(2,2n-2)$. 
This means that any fixed point $x$ approximates other points of measure $\exp\{-n \log 2\}$ in such a way that 
$(\star)$ is at least half  as large as $(*)$ . 
So, assuming that there exists an input $xx^*$ which gives a large value in $(*)$ 
we get a contradiction because if we take random quantum channels $(\star)$ is likely to be small 
with the large deviation bound as in Theorem \ref{deviation}; 
we just set parameters to get proper constants which result in a contradiction. 

Therefore,  the connection between those two methods can be stated as follows. 
Hastings' method considers how much part of the domain can be approximated by unwanted points to get a contradiction.
Our method focuses on desired points in the domain and use $\epsilon$-net argument to prove the result.
Interestingly, then, both methods result in similar estimates as written above. 
On the other hand, however, our estimate for this approximation in the domain is made only from the norm-like properties whereas 
Hastings' involves probabilistic arguments.

\section{Random unitary channel} \label{sec:RU}
We briefly discuss on a class of channels called \emph{random unitary channels}: 
\be
\Phi(\rho) = \frac 1 k \sum_{i=1}^k U_i \rho U_i^* 
\ee
where $U_i \in \mathcal U(n)$. 
Through these channels, input states will be rotated by $U_i$ with equal probability. 
To construct random channels in this class, we take $U_i$ with respect to the Haar measure independently. 

It may seem obvious that 
additivity violation  holds for this class too,
 but since this class forms a measure-zero set in the general channels it is not rigorously obvious. 
However, this class of channels are very close to the one considered in Hastings' paper \cite{Hastings2009},
so additivity violation for this class may be deduced from it. 
If one wants to use our method, one can use the measure concentration argument on
\be
S_{\C^n} \times \cdots \times S_{\C^n}
\ee 
See 6.5.2 of \cite{MilmanSchechtman1985} for details
where one can find that this product space forms a normal L\'evy family.

\section{Regularized minimum output entropy} \label{sec:reg}
We define the regularized minimum output entropy of channels as follows. 
\be
\bar S_{\min}(\Phi) = \lim_{n \to \infty } \frac 1 n S_{\min} \left(\Phi^{\otimes n}\right)
\ee
The limit exists from the following property:
\be
S_{\min} \left(\Phi^{\otimes(m+n)}\right) \leq S_{\min} \left(\Phi^{\otimes m}\right) +S_{\min} \left(\Phi^{\otimes n}\right)  
\ee
We think that it may be a good idea to investigate the following additivity question:
\be\label{reg?}
\bar S_{\min}(\Phi \otimes \Omega) \stackrel{?}{=} \bar S_{\min}(\Phi) + \bar S_{\min}(\Omega)  
\ee
to understand better the question of additivity of classical capacity  in \eqref{open}.
This is because the former problem can be analyzed by eigenvalues of output states 
while the latter needs the geometry of output states. 
In fact,  this eigenvalue approach lead us to discovery of additivity violation with a help of random matrix theory;
we could have proved additivity violation of Holevo capacity somehow but it did not happen.  
This is why, we suggest that the question \eqref{reg?} should be asked first.
In fact, Theorem \ref{C-S} supports this idea. 
We state and prove a widely known fact which 
shows a relation between $C(\cdot)$ and $\bar S_{\min}(\cdot)$ by extending the proof method in \cite{Shor2004}.
\begin{thm}\label{C-S} 
Additivity violation of regularized minimum output entropy will imply 
additivity violation of classical capacity.
\end{thm}
\begin{proof}
Suppose there are some channels $\Phi$ and $\Omega$ such that 
\be
\lim_{n \rightarrow \infty}\frac{1}{n} S_{\min} \left( \Phi^{\otimes n} \otimes \Omega^{\otimes n} \right) 
< \lim_{n \rightarrow \infty}\frac{1}{n} S_{\min} \left( \Phi^{\otimes n}  \right) 
+ \lim_{n \rightarrow \infty}\frac{1}{n} S_{\min} \left( \Omega^{\otimes n} \right) 
\ee
Then, by Lemma \ref{channel extension}, 
there are channels $\tilde \Phi$ and $\tilde \Omega$ such that
\be
\log (k_1 k_2)  - \lim_{n \to \infty}\frac{1}{n} \chi \left( \tilde\Phi^{\otimes n} \otimes \tilde\Omega^{\otimes n} \right) 
< \log k_1 - \lim_{n \to \infty}\frac{1}{n} \chi \left( \tilde\Phi^{\otimes n}  \right) 
+ \log k_2 -  \lim_{n \to \infty}\frac{1}{n} \chi \left( \tilde\Omega^{\otimes n} \right) 
\ee
where $k_1$ and $k_2$ are output dimensions of $\Phi$ and $\Omega$, respectively. 
\end{proof}

To complete the above proof  we need to show Lemma \ref{channel extension}.
To this end, we introduce the following definitions.
For the additive group $\mathbb{Z}_k = \{0,1, \ldots,k-1\}$ we define the discrete Weyl operators on $\C^k$:
\be \label{Weyl}   
W_z = U^x V^y \qquad \text{where} \quad z=(x,y) \in \mathbb{Z}_k \times \mathbb{Z}_k
\ee
Here, $U$ and $V$ are defined as 
\be
U e_r= e_{r+1}\quad\text{and} \quad V e_r= \exp\{ 2\pi {\rm i} r/k \} \cdot e_r 
\qquad (r = 0, \ldots, k-1)
\ee 
where $\{e_0, \ldots, e_{k-1} \}$ is the canonical basis of $\C^k$.

\begin{lemma}\label{channel extension}
Take two channels 
\be
\Phi : L \left(\C^l \right) \rightarrow L \left(\C^k\right)
\qquad \text{and}\qquad
\Omega: L \left(\C^{l^\prime}\right) \rightarrow L \left(\C^{k^\prime}\right)  
\ee
then there exist channels $\tilde\Phi$ and $\tilde\Omega$ such that
\be
\chi \left(\tilde\Phi^{\otimes m} \otimes \tilde\Omega^{\otimes n}\right)
= \log (k^m (k^{\prime})^n) -  S_{\min} \left(\Phi^{\otimes m} \otimes \Omega^{\otimes n}\right)
\qquad \text{for} \quad\forall m,n \in \N \cup \{0\}
\ee
\end{lemma}
\begin{proof}
Define a channel $\tilde \Phi : L (\C^{k^2} \otimes \C^l) \rightarrow L (\C^{k})$ such that
\be
\tilde \Phi (\rho) 
= \sum_{z \in \Z_k \times \Z_k}  W_z \Phi\left(( e_z^* \otimes I ) \rho (e_z  \otimes I)\right) W^*_z 
\ee
Here,   $e_z= e_x \otimes e_y $ so that $\{e_z\}$ is the canonical basis of $\C^{k^2}=\C^k \otimes \C^k$. 
We also define $\tilde\Omega$ in a similar way. 

Suppose 
\be
S_{\min} \left(\Phi^{\otimes m} \otimes \Omega^{\otimes n}\right)
= S\left(\Phi^{\otimes m} \otimes \Omega^{\otimes n}(\rho_0)\right)
\ee
for some $\rho_0 \in L((\C^l)^{\otimes m} \otimes (\C^{l^\prime})^{\otimes n})$.
Then, think of states 
\be
E\left(z^{(m)}\right) \otimes E\left(z^{\prime (n)}\right)  \otimes \rho_0   
\ee
Here, $z^{(m)} = (z_1,\ldots , z_m)$ are strings of $\Z_k \times \Z_k$ of length $m$, and
$z^{\prime (n)} = (z^\prime_1, \ldots, z^\prime_n)$
 of $\Z_{k^\prime} \times \Z_{k^\prime}$ of length $n$ so that
\be
E\left(z^{(m)}\right)  
 = e_{z_1}e_{z_1}^* \otimes \cdots \otimes e_{z_m} e_{z_m}^*
\ee
and $E(z^{\prime(n)})$ is defined similarly. 
Note that combinations of these two strings amount to $k^{2m} (k^{\prime})^{2n}$. 
Then, we claim that the ensemble of  states made from all the possible strings with equal probability leads us to our conclusion. 
First,
\be
S\left(\frac{1}{k^{2m} (k^{\prime})^{2n}} \sum_{(z^{(m)},z^{\prime (n)})} 
\tilde\Phi^{\otimes m} \otimes \tilde\Omega^{\otimes n} \left(E\left(z^{(m)}\right) \otimes E\left(z^{\prime (n)}\right)  \otimes \rho_0   \right) \right) 
= \log (k^m (k^{\prime})^n)
\ee

Secondly,
for each $(z^{(m)},z^{\prime (n)})$, 
\be
S\left(
\tilde\Phi^{\otimes m} \otimes \tilde\Omega^{\otimes n}  \left(E\left(z^{(m)}\right) \otimes E\left(z^{\prime (n)}\right)  \otimes \rho_0   \right)  \right)
&=&
S\left(
\Phi^{\otimes m} \otimes \Omega^{\otimes n} (\rho_0) \right) \\
&=&S_{\min} \left(\Phi^{\otimes m} \otimes \Omega^{\otimes n} \right)
=S_{\min} \left(\tilde\Phi^{\otimes m} \otimes \tilde\Omega^{\otimes n} \right)
\ee
Here, the last equality is from the concavity of $S(\cdot)$.
Therefore,
\be
\chi\left(\tilde\Phi^{\otimes m} \otimes \tilde\Omega^{\otimes n} \right)
= \log (k^m (k^{\prime})^n) - S_{\min} \left(\tilde\Phi^{\otimes m} \otimes \tilde\Omega^{\otimes n} \right)
\qquad \text{for} \quad \forall m,n \in \N\cup\{0\}
\ee
Indeed, the RHS is the upper bound for the Holevo capacity, which has been achieved by the ensemble. 
\end{proof}

\section{Concluding remark} 
In this paper, we developed concise proofs on additivity violation of minimum output entropy of quantum channels. 
In  regimes 
where dimensions of input and output are proportional to each other and
proportionally larger than or equal to square of dimension of environment,
we proved that asymptotically the violation is typical.   
Nevertheless, there are some interesting questions left. 
1)  Is the pair - a quantum channel and its complex conjugate - the best for the violation? 
2)  Is the violation a phenomenon for bipartite systems? 
 Through the project in \cite{CollinsFukudaNEchita2012}, I feel that the first question is true for the random quantum channels.  
For the second question, weak form of additivity is proven in \cite{Montanaro2013}.
 Also, Hastings conjectured in \cite{Hastings2009} that the additivity holds for quantum channels of the form $\Phi \otimes \bar\Phi$.
A positive mathematical evidence for this conjecture was found in \cite{FukudaNechita2012}. 
These results naively   suggest that additivity violation may be a concept for bipartite systems.
More researches should be done to answer these questions.

\appendix
\section{Results from asymptotic geometric analysis}
In this appendix, we collect results in asymptotic geometric analysis which we need. 
We refer interested readers to \cite{MilmanSchechtman1985}.

Let $X$ be a space with metric $\rho$ and Borel probability measure $\mu$.
Then, $(X_r,\rho_r,\mu_r)$ with $r \in \mathbb N$ is called a \emph{normal L\'evy family} 
with constants $c_1,c_2>0$ if
\be
1-\mu(A_r^{\varepsilon}) \leq c_1 \exp \{-c_2 \varepsilon^2 r\}
\ee  
for all $A_r^{\varepsilon}$ with  $\varepsilon >0$ and $r \in \mathbb N$. 
Here, $A_r^{\varepsilon} \subseteq X$ is defined for Borel sets $A_r\subseteq X_r$ with $\mu(A_r)\geq \frac{1}{2}$ 
in the following way:
\be \label{e-expand}
A_r^\varepsilon = \{x\in X_r : \rho (x,A_r) \leq \varepsilon\}
\ee
The unit spheres form a normal L\'evy family; see, for example,  2.2 of \cite{MilmanSchechtman1985}: 
\begin{thm} \label{MS1}
The unit sphere $S^{r+1} \subset  \mathbb R^{r+2}$ with 
the geodesic metric and the uniform measure is a normal L\'evy family with $c_1=\sqrt{\frac{\pi}{8}}$ and $c_2 = \frac{1}{2}$.  
\end{thm} 
Based on this result, we state L\'evy's lemma \cite{Levy1951} in our view 
that behavior of the Lipschitz constant outside $A_r^\varepsilon$ does not matter:
\begin{thm}[L\'evy's Lemma in our view] \label{Levy}  
For $S_{\C^k \otimes \C^n} = S^{2kn-1} \subset \R^{2kn}$ with $k \in \N$ fixed, 
take a sequence of continuous functions $f_n: S_{\C^k \otimes \C^n} \rightarrow \mathbb R$  in the Hilbert-Schmidt norm,
and define $A_n=\{x \in S_{\C^k \otimes \C^n} : f_n(x) \leq \mathrm{med}(f_n) \}$.
Suppose there exist $\varepsilon>0$  and $L>0$  such that 
the Lipschitz constant of $f_n$ is upper-bounded by $L$ on $A_n^{\varepsilon} \setminus A^{\circ}$. 
Then,
\be \label{eq:Levy} 
\mu \left \{ x\in S_{\C^k \otimes \C^n}: 
f_n(x)  >  \mathrm{ med} (f_n) + \varepsilon L  \right\} \leq \sqrt{\frac{\pi}{8}} \exp \left\{ - \varepsilon^2  \left(kn - 1\right)\right\}
\ee 
\end{thm}  
\begin{proof}
The proof is identical to the one for the usual L\'evy's lemma:
\be
\mu \left\{ x \in S_{\C^k \otimes \C^n}: f_n(x)  >   \mathrm{med} (f_n) + \varepsilon L  \right\} 
\leq \mu \left(S^{2kn-1} \setminus A_n^{\varepsilon} \right)  = 1 - \mu \left(A_n^{\varepsilon} \right)
\ee
Indeed, $x\in A_n^{\varepsilon}$ implies $f_n(x) \leq \mathrm{med}(f_n) + \varepsilon L$. 
Note that we switched metric from the geodesic distance to the Hilbert-Schmidt distance where the former is always larger than the latter. 
\end{proof}

\section*{Acknowledgment}
This research was initiated after the author gave a talk on additivity violation at the workshop
 ``Probabilistic Methods in Quantum Mechanics'' at Institut Camille Jordan of Lyon 1 in  France.  
The author thanks Beno\^it Collins for the invitation to AIMR of Tohoku University in Sendai Japan, where the author worked on this project 
and had useful discussions with him. 
Michael Wolf is thanked for his support. 
Roman Vershynin is thanked for useful references.
This research was financially supported by the CHIST-ERA/BMBF project CQC. 
Masahito Hayashi and an anonymous referee of Communications in Mathematical Physics gave useful comments on the first version. 

\bibliographystyle{alpha}
\bibliography{violation-net-bib} 

\end{document}